\documentclass[pra,twocolumn, preprintnumbers,amsmath,amssymb,superscriptaddress]{revtex4}

\usepackage{graphicx, color, graphpap}
\usepackage{amsmath}
\usepackage{enumitem}
\usepackage{amssymb}
\usepackage{amsthm}
\usepackage{pstricks}
\usepackage{float}
\usepackage{multirow}
\usepackage{hyperref}

\long\def\ca#1\cb{} 

\newcommand{\qmu}{q_{\textup{MU}}}
\newcommand{\rhall}{r_{\textup{H}}}
\newcommand{\rgal}{r_{\textup{G}}}
\newcommand{\cmax}{c_{\textup{max}}}
\newcommand{\rr}{r}
\newcommand{\qq}{q}

\newcommand{\vertiii}[1]{{\left\vert\kern-0.25ex\left\vert\kern-0.25ex\left\vert #1 
    \right\vert\kern-0.25ex\right\vert\kern-0.25ex\right\vert}}

\newcommand{\ket}[1]{|#1\rangle}               
\newcommand{\colo}{\,\hbox{:}\,}              
\newcommand{\bra}[1]{\langle #1|}              
\newcommand{\dya}[1]{\ket{#1}\!\bra{#1}}
\newcommand{\dyad}[2]{\ket{#1}\!\bra{#2}}        
\newcommand{\ip}[2]{\langle #1|#2\rangle}      

\newcommand{\EC}{\mathcal{E}}

\newcommand{\HC}{\mathcal{H}}
\newcommand{\IC}{\mathcal{I}}

\newcommand{\XC}{\mathcal{X}}

\newcommand{\Tr}{{\rm Tr}}

\renewcommand{\geq}{\geqslant}
\renewcommand{\leq}{\leqslant}
\newcommand{\mte}[2]{\langle#1|#2|#1\rangle }
\newcommand{\mted}[3]{\langle#1|#2|#3\rangle }

\newcommand{\ot}{\otimes}
\newcommand{\ad}{^\dagger}

\newcommand*{\id}{\openone}

\newcommand{\rhoh}{\widehat{\rho}}
\newcommand{\rhob}{\overline{\rho}}

\newcommand{\dl}{\delta }
\newcommand{\Dl}{\Delta}

\renewcommand{\th}{\theta } 

\newcommand{\lm}{\lambda }

\newcommand{\sg}{\sigma }

\newcommand{\ddd}{d }

\newtheoremstyle{example}{\topsep}{\topsep}%
{}
{}
{\bfseries}
{.}
{   }
{\thmname{#1}\thmnumber{ #2}}
\theoremstyle{example}
\newtheorem{example}{Example}

\theoremstyle{plain}
\newtheorem{theorem}{Theorem}
\newtheorem{lemma}[theorem]{Lemma}
\newtheorem{corollary}[theorem]{Corollary}

\theoremstyle{definition}

\begin{document}

\title{Improved entropic uncertainty relations and information exclusion relations}

\author{Patrick J. Coles}
\affiliation{Centre for Quantum Technologies, National University of Singapore, 2 Science Drive 3, 117543 Singapore}

\author{Marco Piani}
\affiliation{Institute for Quantum Computing and Department of Physics and Astronomy, University of Waterloo, N2L3G1 Waterloo, Ontario, Canada}

\begin{abstract}
The uncertainty principle can be expressed in entropic terms, also taking into account the role of entanglement in reducing uncertainty. The information exclusion principle bounds instead the correlations that can exist between the outcomes of incompatible measurements on one physical system, and a second reference system. We provide a more stringent formulation of both the uncertainty principle and the information exclusion principle, with direct applications for, e.g., the security analysis of quantum key distribution, entanglement estimation, and quantum communication. We also highlight a fundamental distinction between the complementarity of observables in terms of uncertainty and in terms of information.
\end{abstract}


\maketitle

\section{Introduction}

A fundamental trait of quantum mechanics is the unavoidable uncertainty associated with measuring incompatible observables, i.e., the so-called uncertainty principle, which dates back to Heisenberg~\cite{Heisenberg}. Kennard~\cite{kennard1927quantum} formalised Heisenberg's original ideas in an uncertainty relation involving the products of standard deviations of the position $y$ and momentum $p_y$ observables, with the well-known inequality $\Dl y \Dl p_y \geq \hbar / 2$. Robertson \cite{Robertson} generalised this to arbitrary Hermitian observables $X$ and $Z$ and found the uncertainty relation $\Dl X \Dl Z \geq \frac{1}{2} | \bra{\psi} [ X, Z ] \ket{\psi} |$. From a conceptual point of view, though, standard deviation is an inadequate measure of uncertainty, when the latter is understood in terms of (lack of) knowledge of ``which outcome'' of a measurement, rather than in terms of the value of the outcome. Also, the right-hand side (r.h.s.) of Robertson's relation gives a trivial bound for states $\ket{\psi}$ that have zero expectation of the commutator, even if $\ket{\psi}$ is not a common eigenstate of $X$ and $Z$. It has thus been proposed to use the \textit{entropy} of the probability distribution of the outcomes as the measure of  uncertainty~\footnote{See \cite{EURreview1} for a historical review, and \cite{deutsch,EURreview2} for reasons why standard deviation is an inadequate uncertainty measure.}.

The best known \textit{entropic uncertainty relation} is probably the one by Maassen and Uffink \cite{MaassenUffink}. They proved that, for any state $\rho_A$ of a quantum system $A$ with a finite dimension $d = \dim (\HC_A)$, it holds
\begin{equation}
\label{eqn23777a}
H(X)+H(Z) \geq q_{\textup{MU}},
\end{equation}
where $X = \{ \ket{x_j}\} $ and $Z = \{\ket{z_k} \}$ indicate here orthonormal bases on $\HC_A$, and $H(X) = - \sum_j p^x_j \log_2 p^x_j $ is the Shannon entropy of the probability distribution $\{p^x_j = \mte{x_j}{\rho_A}\}$ (similarly for $H(Z)$). The r.h.s.\ of \eqref{eqn23777a} measures the strength of the knowledge tradeoff:  the sum of the ``ignorance'' (as measured by entropy) about $X$ and $Z$ cannot be smaller than
\begin{equation}
\label{eqn23777b}
\qmu= \log_2(1/\cmax),\quad \cmax = \max_{j,k} c_{jk},\quad c_{jk}=|\ip{x_j}{z_k}|^2.
\end{equation}
One has $\qmu=0$ (i.e., $\cmax=1$) if and only if (iff) $X$ and $Z$ share a basis element, while $\qmu$ is maximal, $\qmu=\log_2 d$ (i.e., $\cmax= 1/d$), iff $X$ and $Z$ are fully complementary, with $c_{jk}=1/d$ for all $j,k$.

The uncertainty principle inspired the original proposal for quantum cryptography \cite{Wiesner}. However, the uncertainty relations known at the time did not take into account the possibility for an eavesdropper to have quantum correlations, i.e., entanglement~\cite{HHHH09}, with the system being measured. Hence, those relations could not be directly used to prove cryptographic security.  Berta et al.\ \cite{BertaEtAl} filled such a gap, generalizing the uncertainty relation \eqref{eqn23777a} to take into account the possible use of a quantum memory. The latter would allow Bob, who is supposed to have access to a quantum system $B$ that may be entangled to Alice's system $A$, to violate \eqref{eqn23777a}~\cite{BertaEtAl}. Berta et al. showed that nonetheless, for any bipartite state $\rho_{AB}$, Bob's uncertainty about the result of measurements in the $X$ and $Z$ bases on Alice's system is bounded by
\begin{equation}
\label{eqn23778}
H(X|B)+H(Z|B) \geq \qmu +H(A|B),
\end{equation}
where $H(A|B)=H(\rho_{AB}) - H(\rho_B)$ is the \emph{conditional von Neumann entropy}, with $H(\sg) = -\Tr(\sg\log_2 \sg)$ the von Neumann entropy, and $\rho_B$ the reduced state of $\rho_{AB}$ on $B$.  $H(X|B)$ can be interpreted as Bob's ignorance about the result of Alice's measurement of $X$ on $A$, given that Bob has access to the system $B$ (similarly for $H(Z|B)$)~\footnote{See Cor.~\ref{thmMain230a} for a precise definition of $H(X|B)$.}. The two terms $H(X|B)$ and $H(Z|B)$ are non-negative since they represent classical uncertainties, but $H(A|B)$ can be negative if $\rho_{AB}$ is entangled~\cite{HHHH09}, so that the effect of entanglement is to weaken the knowledge tradeoff. While equation~\eqref{eqn23778} reduces to \eqref{eqn23777a} when $B$ is a trivial system, if $AB$ are maximally entangled, $\rho_{AB}=\ket{\phi}\bra{\phi}$, $\ket{\phi}=(1/\sqrt{d})\sum_i\ket{i}\ket{i}$, we have $H(A|B)= -\log_2 d \leq - \qmu$ independently of $X$ and $Z$, and the r.h.s.\ of \eqref{eqn23778} gives a \textit{trivial} bound on Bob's uncertainty. The generality of \eqref{eqn23778} opens up a range of applications, e.g., in entanglement witnessing \cite{BertaEtAl,PHCFR, LXXLG} and in the security analysis of quantum key distribution \cite{TomRen2010, TLGR}.

\subsection{Summary of results}

One main result of this article is to improve the bound in \eqref{eqn23778} by replacing $\qmu$ with a larger parameter almost always strictly greater than $\qmu$. Another result is the improvement of  Hall's ``information exclusion principle" \cite{Hall1}, which regards the mutual information between the outcomes of measurements on one physical system, and a second system correlated with the first system. Mutual information is a measure of correlations, and is the central quantity in, e.g., communication theory~\cite{CvTh06}. It quantifies the number of bits of information gained---equivalently, the reduction of ignorance---about $X$ when given access to $Y$, and can indeed be defined as $I(X\colo Y) = H(X)- H(X|Y)$. Hall's idea was essentially to reformulate the uncertainty principle in terms of mutual information, as follows. Let $X$ and $Z$ be two orthonormal bases on system $A$, and let $Y$ be a classical register that may be correlated to $A$. Then
\begin{gather}
\label{eqn23782}
I(X\colo Y)+I(Z\colo Y) \leq \rhall,\quad\rhall = \log_2 (d^2 \cdot \cmax ).
\end{gather}
Hall's bound says that one cannot probe the register $Y$ in order to obtain complete information about both the $X$ and $Z$ observables, if these two observables have a small value of $\cmax$ (defined in \eqref{eqn23777b}).
Bounds on the sum of complementary information terms have been called \textit{information exclusion relations} \cite{Hall1, HallPRA1997, ColesEtAl}. They have not been studied as much as uncertainty relations~\footnote{Except in applications involving transmission over quantum channels \cite{ChristWinterIEEE2005, ColesEtAl}, where information is a more a natural quantity than uncertainty.}, and the best known information exclusion relation, Eq. \eqref{eqn23782}, is actually not a very strong bound, as pointed out by Grudka et al.\ \cite{GrudkaEtAl2012}. Grudka et al. have attempted to remedy this by conjecturing a stronger information exclusion relation. They found numerical evidence, and proved analytically in some special cases, that
\begin{equation}
\label{eqn23782bb}
I(X\colo Y)+I(Z\colo Y) \leq \rgal,\quad \rgal = \log_2 \left(\ddd \cdot \sum_{\text{d largest}} c_{jk}\right),
\end{equation}
with the sum over the largest $\ddd$ terms of the matrix $[c_{jk}]$ (again, see \eqref{eqn23777b}). Since $\sum_{\text{d largest}} c_{jk} \leq d\cdot \cmax$ , we have $\rgal \leq \rhall$ (potentially with strict inequality) and, if true, \eqref{eqn23782bb} would be an improvement over Hall's bound. In what follows we shall actually prove a \emph{stronger} version of Grudka et al.'s conjecture. Furthermore we will extend it to the much more general case of quantum memory, where $Y$ is replaced by a general quantum system.

Besides improving both the uncertainty relation~\eqref{eqn23778} and the information exclusion relation~\eqref{eqn23782}, this article provides the insight that the complementarity of uncertainty (i.e., a limit on the knowledge about the outcomes of complementary observables) and the complementarity of information (i.e., a limit on the correlations between the outcomes of complementary observables and some external system) differ both conceptually and practically. Hence, from the quantitative point of view, we should not expect to have the same complementarity factor appearing in uncertainty relations and information exclusion relations. What makes Hall's bound weak is the use in \eqref{eqn23782} of the same parameter $\cmax$ as in \eqref{eqn23777a}.

In what follows, we first give a simplified presentation of our results in Secs.~\ref{sct2} and \ref{sct3} and then discuss their implications in Secs.~\ref{sct4} and \ref{sct5}. We then give a more detailed presentation, discussing the generalisation of our results for arbitrary positive operator valued measures (POVMs) in Sec.~\ref{apsct1}, and giving more details on our state-independent approach in Sec.~\ref{apsct2}. The main technical proofs are given in the Appendix.

\section{Improved uncertainty relation}\label{sct2}

Our main technical result is an entropic uncertainty relation that, much like \eqref{eqn23778}, accounts for the possible reduction of Bob's uncertainty about Alice's system thanks to the entanglement between systems $A$ and $B$. Before presenting our strongest result, we focus on a simple corollary that gives intuition about the nature of our improvement (see Appendix~\ref{app:proofcor} for the proof).

\begin{corollary}
\label{thmMain230a}
For any bipartite state $\rho_{AB}$, and any orthonormal bases $X = \{ \ket{x_j}\} $ and $Z = \{\ket{z_k} \}$ on $\HC_A$,
\begin{equation}
\label{eqn23779aaaaa}
H(X|B)+H(Z|B) \geq q' + H(A|B),
\end{equation}
where $H(X|B)  = H(\rho_{XB})-H(\rho_B)$, with $\rho_{XB}=(\XC \ot \IC)(\rho_{AB})$ and $\XC(\cdot)= \sum_j \dya{x_j}(\cdot)\dya{x_j}$ (similarly for $H(Z|B)$), and
\begin{equation}
\label{eqn23779af}
q' = \qmu+ \frac{1}{2}(1-\sqrt{\cmax})\log_2\frac{\cmax}{c_2},
\end{equation}
where $c_2$ is the second largest entry of the matrix $[c_{jk}]$.
\end{corollary}
Notice that, for small $\cmax$, like in the case of almost complementary $X$ and $Z$, one has $q' \approx \log_2(1/\sqrt{\cmax c_2})$---to be compared with $\qmu= \log_2 (1/\cmax)$. So our bound nicely captures the importance of both $\cmax$ and $c_2$, i.e., takes into account more information about the relation between the two bases. Clearly $q'\geq \qmu$ in general. Furthermore $q'>\qmu$ iff there is exactly one pair $(\hat{j},\hat{k})$ such that $\cmax= c_{\hat{j}\hat{k}}$, with $\cmax<1$. In the special case where the system $A$ is a qubit, it is immediate to check that necessarily $\cmax =c_2$, hence $q'=\qmu$. However, for $d\geq 3$, we have $q' >\qmu$ for \emph{almost all} pairs of bases $(X,Z)$. Indeed, in $d\geq 3$ a typical unitary---seen here as the unitary that connects the two bases, i.e., $X=\{\ket{x_j}\}=\{U\ket{z_j}\}$---has $c_2 <  \cmax < 1$, see Sec.~\ref{apsct2.2}. We remark that, even for the simple improvement provided by Corollary~\ref{thmMain230a}, the gap between $q'$ and $\qmu$ can become arbitrarily large. In Sec.~\ref{apsct2.3}, we give an example where the gap $q'-\qmu$ diverges as the logarithm of the dimension $d$ of $A$~\footnote{This is the largest possible dependence on $d$, since $H(X|B)+H(Z|B)\leq 2 \log_2 d$.}.

We now state our main technical result, from which all of our other relations follow~(see Appendix~\ref{app:proofthm} for the proof). We first replace the bound $\qmu$ in \eqref{eqn23778} with a \emph{state-dependent} bound $q(\rho_A)$, and then define a new state-independent bound. 

\begin{theorem}
\label{thmMain230}
For any bipartite state $\rho_{AB}$, and any orthonormal bases $X = \{ \ket{x_j}\} $ and $Z = \{\ket{z_k} \}$ on $\HC_A$,
\begin{equation}
\label{eqn23779}
H(X|B)+H(Z|B) \geq q(\rho_A) + H(A|B),
\end{equation}
where, from $c_{jk}$ in \eqref{eqn23777b}, we define
\begin{subequations}
\label{eq:qprime}
\begin{align}
q(\rho_A) &= \max\{q(\rho_A, X,Z), q(\rho_A, Z,X)\},\\
q(\rho_A, X,Z) &= \sum_j p^x_j \log_2 (1/\max_k c_{jk} ),\\
q(\rho_A, Z,X) &= \sum_k p^z_k \log_2 (1/\max_j c_{jk} ).
\end{align}
\end{subequations}
Hence, the following state-independent bound holds:
\begin{equation}
\label{eq:stateindep}
H(X|B)+H(Z|B) \geq q+ H(A|B),\quad q = \min_{\rho_A} q(\rho_A).
\end{equation}
\end{theorem}
It is clear that $q(\rho_A)\geq \qmu$, since averaging over $j$ or $k$ gives a larger value than minimising. For $A$ a qubit ($\ddd=2$), we have that $\max_k c_{jk} $ is independent of $j$ and hence $q(\rho_A)=\qmu$. But $q(\rho_A)\geq q'$ (see the proof of Cor.~\ref{thmMain230a}), so that, for $\ddd \geq3$, $q(\rho_A)\geq q' > \qmu$ for all states, for almost all choices of $X$ and $Z$.
Hence the lower bound $q$ of \eqref{eq:stateindep} is an improvement over $\qmu$.
By using the minimax theorem, see Sec.~\ref{apsct2.1}, we obtain
\begin{equation}
\label{eqn237qprime}
q = \max_{ 0 \leq p \leq 1} \lm_{\min} [ \Dl(p)],
\end{equation}
where $\lm_{\min}[\cdot]$ denotes the minimum eigenvalue and $\Dl(p) = p\Dl_{XZ}+(1-p)\Dl_{ZX}$, with $\Dl_{XZ}=\sum_j   \log_2 (1/\max_k c_{jk} )\dya{x_j}$ and $\Dl_{ZX}=\sum_k   \log_2 (1/\max_j c_{jk} )\dya{z_k}$. Thus, computing $q$ can be done by finding the minimum eigenvalue of particular matrices, a straightforward numerical calculation. Furthermore, by setting $p=1/2$ in \eqref{eqn237qprime} one can get a bound \emph{still certified to be at least as large as $q'\geq \qmu$}. In general, we have
\[
q\geq \lm_{\min}[\Dl(1/2)]\geq q' \geq \qmu.
\]
\begin{example}
\label{Ex1}
Let $\ddd=3$, $Z = \{\ket{0},\ket{1},\ket{2}\}$, and $X = \{U \ket{0},U \ket{1},U \ket{2}\}$, with
$$U=\left(\begin{smallmatrix}
     1/\sqrt{3} &  1/\sqrt{3}& 1/\sqrt{3}  \\
    1/\sqrt{2}  & 0 & -1/\sqrt{2}\\
    1/\sqrt{6} & -\sqrt{2/3} &1/\sqrt{6} 
\end{smallmatrix}\right).$$
We have $\qmu = \log_2 (3/2)\approx 0.58$, $q' \approx 0.62$, $\lm_{\min}[\Dl(1/2)]\approx 0.64$, and $q \approx 0.64$. Furthermore, our state-dependent bound is often \textit{much} better than $\qmu$: if the reduced state is maximally mixed then $q(\id/3) = (2/3) \log_2 3 \approx 1.06$, while numerically averaging over all pure states gives $\langle q(\ket{\psi})\rangle_{\ket{\psi}} \approx 1.07$.
\end{example}

Other attempts have been made to strengthen Eq.~\eqref{eqn23778} \cite{PatiEtAlPRA2012, TomamichelThesis2012, TomHang2013}, or the less general relation Eq.~\eqref{eqn23777a} \cite{deVicSanRuizPRA2008,FriedGheorGour2013arxiv,PuchalaEtAl2013}. Refs.~\cite{FriedGheorGour2013arxiv,PuchalaEtAl2013} took a majorisation approach; however, their bounds can be weaker than \eqref{eqn23777a} when $X$ and $Z$ have a large $\qmu$ value. Ref.~\cite{PatiEtAlPRA2012} added a term to the r.h.s.\ of \eqref{eqn23778} that depends on the quantum discord~\cite{ModiEtAlRevModPhys.84.1655} of the state $\rho_{AB}$; that same term (see \cite{PatiEtAlPRA2012}) can be added to the r.h.s.\ of our result \eqref{eqn23779} if one wishes. Ref.~\cite{TomamichelThesis2012} (Ch.~7) replaced $\qmu$ in \eqref{eqn23778} with a state-dependent bound $\hat{q}(\rho_A)$, like we did in \eqref{eqn23779}; however in their case they have $\min_{\rho_A} \hat{q}(\rho_A) = \qmu$, so unlike our result it does not lead to a strengthened state-independent bound.

\section{Improved information exclusion relation}\label{sct3}

As a corollary of \eqref{eqn23779}, we prove Grudka et al.'s conjectured information exclusion relation~\cite{GrudkaEtAl2012}. Furthermore, we actually strengthen their bound and extend it to the case of quantum memory. In order to fully appreciate this, let us first consider the extension of Hall's result to the case of quantum memory, i.e., we replace the classical system $Y$ with a general quantum system $B$. A corollary of \eqref{eqn23778} is:
\begin{equation}
\label{eqn23784}
I(X\colo B)+I(Z\colo B) \leq \rhall - H(A|B).
\end{equation}
Improving \eqref{eqn23782}, this result allows for entanglement between $A$ and $B$. It says that the trade-off in correlations is weakened if $H(A|B)$ is negative, i.e. if $\rho_{AB}$ is strongly entangled. After all, in the maximally entangled case, $I(X\colo B)=I(Z\colo B)=\log_2 \ddd$, so in such a case the bound on the r.h.s.\ must be no smaller than $2 \log_2 \ddd$.

Now consider the following information exclusion relation, a corollary of our uncertainty relation \eqref{eqn23779}.
\begin{corollary}
\label{thmMain231}
For any bipartite state $\rho_{AB}$,
\begin{equation}
\label{eqn23785}
I(X\colo B)+I(Z\colo B) \leq r - H(A|B),
\end{equation}
with
\begin{subequations}
\label{eqn237431}
\begin{align}
r&= \min \{r(X,Z), r(Z,X) \},\\
r(X,Z)&= \log_2 (\ddd \sum_j \max_k c_{jk}),\\
r(Z,X)&= \log_2 (\ddd \sum_k \max_j c_{jk}).
\end{align}
\end{subequations}
\end{corollary}
\begin{proof}
Write $H(X|B) = H(X)- I(X\colo B)$ (similarly for $H(Z|B)$), rearrange \eqref{eqn23779}, and use $H(Z)\leq \log_2 \ddd$ to get
$$I(X\colo B)+I(Z\colo B) \leq \log_2 \ddd + H(X) - q(\rho_A, X,Z) - H(A|B).$$
Now, 
$H(X) - q(\rho_A, X,Z) = \sum_j p^x_j \log_2 (\max_k c_{jk} /p^x_j ) \leq \log_2 ( \sum_j \max_k c_{jk})$,
where we used the concavity of the log. Bringing $\ddd$ inside the log completes the proof. A similar bound holds when interchanging $X$ and $Z$.\end{proof}

This allows us to conclude
\begin{corollary}
\label{thmMain232}
Grudka et al.'s conjecture, \eqref{eqn23782bb}, is true.
\end{corollary}
\begin{proof}
Consider the $d$ different terms $\{\max_k c_{jk}\}_j$ appearing in $r(X,Z)$; these may not be the $d$ largest terms of the matrix $[c_{jk}]$, hence summing over them is smaller than computing $\sum_{\text{d largest}} c_{jk} $. So $r(X,Z) \leq \rgal$ (see \eqref{eqn23782bb}), thus $r \leq \rgal$. Also, if we set $B=Y$, where $Y$ is classical, then we have $H(A|Y) \geq 0$. Combining this with $r \leq \rgal$ and \eqref{eqn23785} proves \eqref{eqn23782bb}. \end{proof}

We emphasise that Eq.~\eqref{eqn23785} goes well beyond Grudka et al.'s conjecture: it strengthens \eqref{eqn23782bb} by replacing $\rgal$ by $r$, and it generalises the result to the case of quantum memory, allowing for arbitrary (possibly non-classical) correlations between $A$ and $B$. In general, we have
\[
r\leq \rgal \leq \rhall.
\]
In the qubit case ($\ddd=2$), we have equality $r = \rgal= \rhall$. To see a case where all three are different, consider the qutrit example given in Ex.~\ref{Ex1}. In this case we have $\rhall= \log_2 6$, $\rgal = \log_2 5$, and $r = \log_2 (9/2)$. Note that $r$ can be calculated analytically given the coefficients $c_{jk}$ of \eqref{eqn23777b}.

\section{Uncertainty versus information} \label{sct4}

One key conceptual insight of our work is to draw a distinction between the complementarity of uncertainty and the complementarity of information. The factor $\cmax$ naturally appears---via $\qmu=\log_2(1/\cmax)$---in uncertainty relations like \eqref{eqn23777a} and \eqref{eqn23778}. But we should not expect it to be the right factor to capture the complementarity of information. While our work shows that uncertainty relations can be improved  by replacing $\qmu$ with $q$ as in \eqref{eq:stateindep}, a much more dramatic improvement is given by replacing $\rhall = \log_2(d^2 \cdot \cmax)$ with $r$, i.e., going from the information exclusion relation~\eqref{eqn23784} to~\eqref{eqn23785}. Indeed, in order to obtain a state-independent bound for uncertainty relations, we must consider the subspace with the least complementarity. On the other hand, in information exclusion relations it is the overall complementarity, i.e.\ with respect to the various subspaces that compose the space, that matters. The reason our approach is better suited to capture information complementarity is that $r$ measures the overall complementarity, averaged over the whole space, of $X$ and $Z$. Notice that to obtain our improved \emph{state-independent} information exclusion relation of Corollary~\ref{thmMain231} we had to tap into the strength of our \emph{state-dependent} uncertainty relation of Theorem~\ref{thmMain230}. Finally, to better appreciate the difference between the complementarity of uncertainty and the complementarity of information, it is instructive to consider the conditions under which our state-independent bounds become trivial, i.e., $q= 0$ and $r= 2 \log_2 d$. Let $U$ be the unitary relating $X$ and $Z$; we have $q= 0 $ iff at least one entry of $U$ has magnitude 1. In contrast, $r = 2 \log_2 d$ iff $U$ is of the form $U = \sum_j e^{i\phi_j} \dyad{P(j)}{j}$ for some permutation function $P$ and phase factors $e^{i\phi_j}$. These are vastly different conditions, with the latter one implying that $U$ must be trivial \textit{over the entire space}, whereas the former condition says that only one row or column of $U$ need be trivial.

\section{Applications} \label{sct5}

The relevance of \eqref{eqn23778} for  witnessing of entanglement (WoE) and  security analysis for quantum key distribution was discussed in \cite{BertaEtAl} and implemented experimentally for WoE in \cite{PHCFR, LXXLG}. Since our bound Eq.~\eqref{eqn23779} is an improvement over \eqref{eqn23778}, it enables a tighter analysis. To use our bound $q(\rho_A)$ the only information about $\rho_A$ needed is the probability distributions $\{p^x_j\}$ and $\{p^z_k\}$. In the case of WoE using the uncertainty relation with quantum memory as in \cite{PHCFR, LXXLG}, Alice already determines these probability distributions experimentally, so no extra effort is needed to use our bound. 

Likewise, our Eq.~\eqref{eqn23785} is relevant to \textit{witnessing of good quantum channels} \cite{ChristWinterIEEE2005, ColesEtAl}. Consider a channel $\EC$ from Alice to Bob. To show that $\EC$ is good, Alice can send the $X$ basis states with equal probability through $\EC$, and Bob measures the output in basis $X_B$. Alice does the same for $Z$ and Bob measures $Z_B$. They compare their results over a classical communication line and estimate $I(X\colo X_B )$ and $I(Z\colo Z_B )$. Then they can lower bound the quantum capacity of $\EC$, denoted $Q(\EC)$, using
\begin{equation}
\label{eqn385792}
Q(\EC) \geq I(X\colo X_B)+I(Z\colo Z_B)  - r,
\end{equation}
which follows from applying \eqref{eqn23785} to $\rho_{AB} = (\IC \ot \EC)(\dya{\phi})$ where $\ket{\phi}$ is maximally entangled, and using $Q(\EC)\geq -H(A|B)$ \cite{Lloyd97}. Thus, showing that $\EC$ has a positive quantum capacity amounts to showing that the r.h.s.\ of \eqref{eqn385792} is positive.

Closely related to quantum cryptography are ideas of monogamy or decoupling, whereby strong quantum correlations between $A$ and $B$ guarantee weak correlations between $A$ and any third system $C$. Equation~\eqref{eqn23778} has been used \cite{RenesBoileau, ColesEtAl} to give sufficient conditions for which $C$ is decoupled from $A$, in terms of Bob's uncertainty about $X$ and $Z$. Our results, \eqref{eqn23779} and \eqref{eqn23785}, allow these quantitative statements of monogamy to be tightened.

\section{Generalisation to POVMs}\label{apsct1}

\subsection{Results in tripartite form}\label{apsct1.1}

Our previous results can be rewritten in a form that considers a tripartite state on $ABC$ rather than a bipartite state on $AB$. The tripartite formulation is equivalent to the bipartite one, i.e., one formulation implies the other \cite{BertaEtAl,ColesEtAl}. In what follows, we state the tripartite formulation of our results since this form allows us to generalise our results to POVMs in a straightforward way.

Our first main result was Eq.~\eqref{eqn23779}. This says that, for any tripartite state $\rho_{ABC}$ and any orthonormal bases $X = \{ \ket{x_j}\} $ and $Z = \{\ket{z_k} \}$ on $\HC_A$,
\begin{equation}
\label{eqn23779app}
H(X|B)+H(Z|C) \geq \qq (\rho_A)
\end{equation}
where $\qq (\rho_A)$ was defined in \eqref{eq:qprime}. Notice that the term $H(A|B)$ that appeared in (8) has now disappeared since we have changed $H(Z|B)$ to $H(Z|C)$. 

Our second main result was Eq.~\eqref{eqn23785}. It says that, for any tripartite state $\rho_{ABC}$ and any orthonormal bases $X = \{ \ket{x_j}\} $ and $Z = \{\ket{z_k} \}$ on $\HC_A$,
\begin{equation}
\label{eqn23785app}
I(X\colo B)+I(Z\colo C) \leq \rr ,
\end{equation}
where $\rr$ was defined in \eqref{eqn237431}.

We used our second result to prove a conjecture by Grudka et al. \cite{GrudkaEtAl2012}, which strengthened Hall's information exclusion principle \cite{Hall1}. Hall's scenario considered the case where $Y$ is a classical register and we want to bound the sum $I(X\colo Y)+I(Z\colo Y)$. Our second result implied the following bound on this sum:
\begin{equation}
\label{eqn2958383app}
I(X\colo Y)+I(Z\colo Y) \leq \rr ,
\end{equation}
which in turn implied Grudka et al.'s conjecture.

In what follows, we will generalise all of these results, Eqs.~\eqref{eqn23779app}, \eqref{eqn23785app}, and \eqref{eqn2958383app}, to the case where $X$ and $Z$ are arbitrary POVMs (assuming they contain a finite number of POVM elements) on system $A$.

\subsection{Notation for POVMs}\label{apsct1.2}

In the general case where $X = \{X_j\}$ and $Z = \{Z_k\}$ are POVMs on $A$, we consider the isometries $V_X\colo \HC_A \to \HC_{XX'A}$ and $V_Z\colo \HC_A \to \HC_{ZZ'A}$ defined by \cite{TomRen2010}
\begin{subequations}
\label{eqnisometry11}
\begin{align}
V_X = \sum_j \ket{j}_X\ot \ket{j}_{X'}\ot \sqrt{X_j},\\
V_Z = \sum_k \ket{k}_Z\ot \ket{k}_{Z'}\ot \sqrt{Z_k},
\end{align} 
\end{subequations}
where $\ket{j}$ and $\ket{k}$ are elements of the standard (orthonormal) basis on the appropriate spaces. For some initial tripartite state $\rho_{ABC}$ we denote the alternative post-measurement states as:
\begin{subequations}
\label{eqnPM11}
\begin{align}
\rhoh_{XX'ABC} = V_X\rho_{ABC}V_X\ad ,\\
\rhob_{ZZ'ABC} = V_Z\rho_{ABC}V_Z\ad .
\end{align}
\end{subequations}
Then we define
\begin{subequations}
\label{eqnCE11}
\begin{align}
H(X|B) = H(\rhoh_{XB})-H(\rho_B) ,\\
H(Z|C) = H(\rhob_{ZC})-H(\rho_C) ,
\end{align}
\end{subequations}
which are the conditional entropies of the classical quantum states $\rhoh_{XB} = \Tr_{X'AC}(\rhoh_{XX'ABC})$ and $\rhob_{ZC} = \Tr_{Z'AB}(\rhob_{ZZ'ABC})$, respectively. For example, notice that we can write
\begin{align}
\rhoh_{XB} &= \sum_j \dya{j}_X \ot \Tr_A(X_j\rho_{AB})\notag \\
\label{eqnXchannel11}&=(\XC \ot \IC)(\rho_{AB})
\end{align}
for the quantum channel $\XC:\rho_A\mapsto\sum_j\dyad{j}{j}_X\Tr(X_j\rho_A)$. Also, we denote the probabilities associated with these two POVMs as $p^x_j = \Tr(X_j \rho_A)$ and $p^z_k = \Tr(Z_k \rho_A)$.

\subsection{Uncertainty relation for POVMs}\label{apsct1.3}

Generalising the results to POVMs essentially amounts to finding an appropriate generalisation of the complementarity factor that appears in our bounds, such as $\qq (\rho_A)$ and $\rr $. In what follows, we will use the factors:
\begin{subequations}
\label{eqn79}
\begin{align}
h_j(X,Z) &=  \Big\| \sum_k Z_k X_j Z_k \Big\|_{\infty} , \\ 
h_k(Z,X) &=  \Big\| \sum_j X_j Z_k X_j \Big\|_{\infty} ,
\end{align} 
\end{subequations}
where the infinity norm (or operator norm) $\| M \|_{\infty}$ is the largest singular value of $M$, or in the case of \eqref{eqn79} it is the largest eigenvalue since the arguments are positive semi-definite matrices. We discuss in the next subsection why we chose this complementarity factor - the reason being that it gives a stronger bound than an alternative, as discussed below.

Now we generalise \eqref{eqn23779app} to the case of arbitrary POVMs with the following result, proved in App.~\ref{app:proofthm2}.

\begin{theorem}
\label{thmUncRelApp}
Let $X = \{X_j\}$ and $Z = \{Z_k\}$ be arbitrary POVMs on $A$. Then for any tripartite state $\rho_{ABC}$,
\begin{equation}
\label{eqn23779appPOVM}
H(X|B)+H(Z|C) \geq \qq (\rho_A)
\end{equation}
where we define
\begin{subequations}
\label{eqnqprimerhoaApp}
\begin{align}
\qq (\rho_A) &= \max\{\qq (\rho_A, X,Z), \qq (\rho_A, Z,X)\},\\
\qq (\rho_A, X,Z) &= - \sum_j p^x_j \log_2 h_j(X,Z),\\
\qq (\rho_A, Z,X) &= - \sum_k p^z_k \log_2 h_k(Z,X).
\end{align}
\end{subequations}
\end{theorem}

Notice that our definition of $\qq (\rho_A)$ reduces to that given in \eqref{eq:qprime} when we specialise to the case of orthonormal bases (in other words, rank-one projective POVMs). This is because, when $Z$ is projective, then $h_j(X,Z) = \max_k \| Z_k X_j Z_k \|_{\infty} $ and further specialising to $X$ and $Z$ being composed of rank-one projectors reduces the formula to $h_j(X,Z) = \max_k c_{jk}$, which is the formula appearing in \eqref{eq:qprime}.

While we have taken the tripartite view to give a simple statement of our results for POVMs, it is possible rewrite \eqref{eqn23779appPOVM} in a bipartite form, using an approach similar to that in \cite{TomamichelThesis2012}. We obtain:
\begin{equation}
\label{eqn23779appPOVMab}
H(X|B)+H(Z|B) \geq \qq (\rho_A) + H(A|B) - f
\end{equation}
where $f:= \min \{H(A|BX)_{\rhoh},H(A|BZ)_{\rhob}\}$, and where $H(A|BX)_{\rhoh}$ and $H(A|BZ)_{\rhob}$ denote the conditional entropies of $\rhoh_{XAB}$ and $\rhob_{ZAB}$, respectively. For the case of orthonormal bases considered earlier, $f=0$.

\subsection{Choice of complementarity factor}\label{apsct1.4}

The following technical lemma, proved in App.~\ref{app:prooflemA}, is relevant to our choice of complementarity factor for POVMs \footnote{Lemma~\ref{thmPOVM3875} was proved in a collaborative discussion with M.~Tomamichel, and approval to publish it in this paper was granted by M.~Tomamichel.}.

\begin{lemma}\label{thmPOVM3875}
Let $\sg$ be an arbitrary operator---that is, an arbitrary square matrix, although we will be interested mostly in the case in which $\sigma$ is positive semidefinite, hence the choice of notation---and let $Z = \{Z_k\}$ be any POVM.
Then
\begin{equation}
\label{eqn237893479}
 \Big\| \sum_k Z_k \sg Z_k \Big\|_{\infty} \leq \max_k \Big\| \sqrt{Z_k} \sg \sqrt{Z_k} \Big\|_{\infty}.
\end{equation}
\end{lemma}

Our choice of complementarity factor was inspired by Refs.~\cite{TomamichelThesis2012,TomHang2013}. In particular, in Chapter 7 of \cite{TomamichelThesis2012} Tomamichel conjectures that, for any two POVMs $X$ and $Z$,
\begin{equation}
\label{eqnTomconj1}
\max_j  \Big\| \sum_k Z_k X_j Z_k \Big\|_{\infty} \leq \max_{j,k} c_{jk},
\end{equation}
where
\begin{equation}
\label{eqnGencjk1}
c_{jk} = \Big\|  \sqrt{Z_k} X_j \sqrt{Z_k} \Big\|_{\infty}= \Big\|  \sqrt{X_j} \sqrt{Z_k} \Big\|_{\infty}^2.
\end{equation}
Clearly our Lemma~\ref{thmPOVM3875} implies Eq.~\eqref{eqnTomconj1} and hence resolves an outstanding conjecture. The reason this conjecture was interesting was because the factors on the left- and right-hand-sides of \eqref{eqnTomconj1} were alternative complementarity factors that could potentially be used as bounds in the uncertainty relation. Indeed the r.h.s.\ of \eqref{eqnTomconj1} was used in several uncertainty relations \cite{ColesEtAl,TomRen2010,ColesColbeckYuZwolak2012PRL}, so proving that the l.h.s.\ of \eqref{eqnTomconj1} is smaller, as we have done here, shows that the l.h.s.\ provides a better bound for POVM uncertainty relations. (This issue is only of concern for general POVMs, since the two factors in \eqref{eqnTomconj1} are equal when $X$ and $Z$ are orthonormal bases.)

This discussion has relevance to the present article since our derived bound in Theorem~\ref{thmUncRelApp} involves quantities $\qq (\rho_A,X,Z)$ and $\qq (\rho_A,Z,X)$ defined in terms of $h_j(X,Z)$ and $h_k(Z,X)$ given in \eqref{eqn79}. But from Lemma~\ref{thmPOVM3875}, these quantities are bounded by
\begin{subequations}
\label{eqnhj5}
\begin{align}
h_j(X,Z) \leq \max_k \big\| \sqrt{Z_k} X_j \sqrt{Z_k} \big\|_{\infty} = \max_k c_{jk},\\
h_k(Z,X) \leq \max_j \big\| \sqrt{X_j} Z_k \sqrt{X_j} \big\|_{\infty} = \max_j c_{jk}.
\end{align} 
\end{subequations}
Hence our bound involving $h_j(X,Z)$ and $h_k(Z,X)$ is stronger than the one obtained from replacing them with the quantities on the right-hand-sides of \eqref{eqnhj5}. This provides justification for our choice of complementarity factor.

\subsection{Information exclusion relation for POVMs}\label{apsct1.5}

Here we use Theorem~\ref{thmUncRelApp} to derive an information exclusion relation that is generalised to the POVM case. Again, we note that the following definition of $\rr $ reduces to that in \eqref{eqn237431} when $X$ and $Z$ are specialised to be orthonormal bases.

\begin{corollary}
\label{thm373}
Let $X = \{X_j\}$ and $Z = \{Z_k\}$ be arbitrary POVMs on $A$. Then for any tripartite state $\rho_{ABC}$,
\begin{equation}
\label{eqn2377439appPOVM}
I (X\colo B)+I(Z\colo C) \leq \rr ,
\end{equation}
where we define
\begin{subequations}
\label{eqnrppdefapp}
\begin{align}
\rr  & = \min\{\rr ( X,Z), \rr (Z,X)\},\\
\rr ( X,Z) &= \log_2 [|Z| \sum_j h_j(X,Z) ],\\
\rr (Z,X) &= \log_2 [|X| \sum_k h_k(Z,X) ].
\end{align}
where $|Z|$ and $|X|$ denote the number of POVM elements.
\end{subequations}
\end{corollary}
\begin{proof}
Write $H(X|B) = H(X)- I(X\colo B)$ and $H(Z|C)= H(Z) - I(Z\colo C)$, then rearrange \eqref{eqn23779appPOVM} and use $H(Z)\leq \log_2 |Z|$ to get
$$I(X\colo B)+I(Z\colo C) \leq \log_2 |Z| + H(X) - \qq (\rho_A, X,Z).$$
Now write 
\begin{align}
H(X) - \qq (\rho_A, X,Z) &= \sum_j p^x_j \log_2 [h_j(X,Z)/p^x_j ] \notag \\
&\leq \log_2 [ \sum_j h_j(X,Z)],
\end{align}
where we used the concavity of the log. Bringing $|Z|$ inside the log completes the proof, and by symmetry the same bound holds where one interchanges $X$ and $Z$.\end{proof}

Finally, we generalise \eqref{eqn2958383app} to the POVM case. The following result is applicable to the same scenario that Hall considered in his information exclusion principle, except we have generalised it to the case where $X$ and $Z$ are POVMs.
\begin{corollary}
Let $X = \{X_j\}$ and $Z = \{Z_k\}$ be arbitrary POVMs on $A$. Let $Y$ be a classical register that may be correlated to $A$, i.e., $\rho_{AY}$ is an arbitrary quantum-classical state. Then,
\begin{equation}
\label{eqn2377439appPOVMhall}
I (X\colo Y)+I(Z\colo Y) \leq \rr 
\end{equation}
where $\rr $ is defined by Eq.~\eqref{eqnrppdefapp}.
\end{corollary}
\begin{proof}
Apply \eqref{eqn2377439appPOVM} to the tripartite state $\rho_{AYY'}$ where system $Y'$ is an exact copy of system $Y$, such that $\rho_{AY'}=\Tr_Y(\rho_{AYY'})$ is of the same form as $\rho_{AY}= \Tr_{Y'}(\rho_{AYY'}) $. (Note: the fact that $Y$ is classical allows us to copy its correlations with $A$.) In this case we have $I (X\colo Y') = I (X\colo Y)$, hence proving \eqref{eqn2377439appPOVMhall}.
\end{proof}

\section{State-independent bound for uncertainty relation}\label{apsct2}

\subsection{Computable expression}\label{apsct2.1}

Now let us consider the state-independent version of our bound, defined by
$$\qq  = \min_{\rho_A} \qq (\rho_A).$$
In Sec.~\ref{sct2} we noted that this bound can be rewritten in an alternative form that may be easier to calculate. Here we derive this alternative form.

Let us first rewrite $\qq (\rho_A)$ as follows:
\begin{align}
\label{eqnqprime121}
\qq (\rho_A) &= \max \{\qq (\rho_A,X,Z),\qq (\rho_A,Z,X)\}\notag\\
 &= \max_{0\leq p\leq 1} \large[p \hspace{2pt}\qq (\rho_A,X,Z) +(1-p)\qq (\rho_A,Z,X)\large]\notag\\
 &= \max_{0\leq p\leq 1} \large[ p\hspace{1pt} \Tr [ \rho_A \cdot \sum_j X_j \log_2 (1/h_j(X,Z))]\notag\\
  &\hspace{10 pt}+(1-p) \Tr [\rho_A \cdot \sum_k Z_k \log_2 (1/h_k(Z,X))] \large]\notag\\
 &= \max_{0\leq p\leq 1} \Tr[ \rho_A \Dl(p)],
\end{align} 
where we define
\begin{align}
\label{eqn237qprime2app}
\Dl(p) &= p\Dl_{XZ}+(1-p)\Dl_{ZX},\\
\Dl_{XZ}&=\sum_j   \log_2 (1/h_j(X,Z))\cdot X_j,\notag\\
\Dl_{ZX}&=\sum_k   \log_2 (1/h_k(Z,X)) \cdot Z_k.\notag
\end{align}
From $h_j(X,Z)\leq 1$ and $h_k(Z,X) \leq 1$, it follows that $\Dl_{XZ}\geq 0$ and $\Dl_{ZX}\geq 0$, and hence $\Dl(p)\geq 0$.

Next, thanks to the linearity in the arguments, we can use the minimax theorem to interchange the min and max in $\qq $ as follows:
\begin{align}
\label{eqnqprime122}
\qq  &= \min_{\rho_A} \max_{0\leq p\leq 1} \Tr[ \rho_A \Dl(p)]\notag\\
 &=  \max_{0\leq p\leq 1} \min_{\rho_A} \Tr[ \rho_A \Dl(p)]\notag\\
 &=  \max_{0\leq p\leq 1} \lm_{\min} [\Dl(p)].
\end{align} 

\begin{figure}
\begin{center}
\includegraphics[width=3.2in]{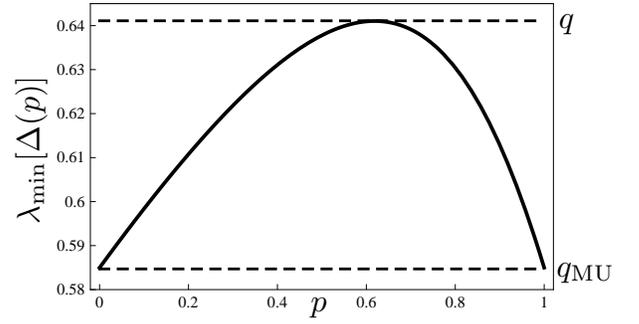}
\caption{Plot of the minimum eigenvalue of $\Dl(p)$ as a function of $p$, for Example~\ref{Ex1}. The maximum in the plot corresponds to $\qq  \approx 0.64$, which is an improvement over the old bound $\qmu \approx 0.58$ corresponding to the value at $p=0$ and $p=1$.}
\label{fgr1}
\end{center}
\end{figure}

The formula in \eqref{eqnqprime122} makes it possible to numerically calculate $\qq $. Given the POVM elements of $X$ and $Z$, it is straightforward to numerically diagonalise $\Dl(p)$ for a fixed $p$; then the maximisation over $p$ can be plotted graphically. For example, Fig.~\ref{fgr1} shows this plot for Example~\ref{Ex1} given in Sec.~\ref{sct2}, yielding a value of $\qq  \approx 0.64$.

It is also worth noticing that, since $\lm_{\min}[\Dl_{XZ}]=\lm_{\min}[\Dl_{ZX}]=\qmu $, Eq.~\eqref{eqnqprime122} is another way of seeing that $\qq  \geq \qmu $. Also, since the smallest eigenvalue satisfies $\lm_{\min}[A+B] \geq \lm_{\min}[A] +\lm_{\min}[B]$ for any two Hermitian matrices~\cite{bhatia1997matrix}, we have that $\qq  = \qmu $ iff the function $\lm_{\min}[\Dl(p) ]$ is independent of $p$ and hence is equal to $\qmu $ for all $p$.

\subsection{Analytical bound}\label{apsct2.2}

While $q$ is our strongest state-independent bound, we can find a slightly weaker state-independent bound $q'$ that is given by a simple, analytical expression and is still an improvement over $\qmu$. In Cor.~\ref{thmMain230a}, we gave the form of $q'$ in terms of the largest and second-largest entries of the matrix $[c_{jk}]$. To state this result for general POVMs, we define $c_{jk}$ according to \eqref{eqnGencjk1}, which reduces to the expression in \eqref{eqn23777b} in the case of orthonormal bases.

Now we generalise Cor.~\ref{thmMain230a} to POVMs as follows, with the proof in App.~\ref{app:proofcor2}.
\begin{corollary}
\label{thmAppQprimebound}
Let $(\hat{j},\hat{k})$ be a pair of indices such that $c_{\hat{j}\hat{k}}=\max_{jk}c_{jk}= \cmax$, where $c_{jk}$ is defined in \eqref{eqnGencjk1}, so that $c_{\hat{j}\hat{k}}= \| \sqrt{X_{\hat{j}}}\sqrt{Z_{\hat{k}}}  \|_\infty^2$. Let $c_2$ be the second-largest entry of the matrix $[c_{jk}]$ (possibly equal to $\cmax $). It holds that $q(\rho_A) \geq q'$ where $q'$ is a state-independent parameter given by
\begin{equation}
\label{eq:lowerboundapp}
\qq ' = \qmu + \frac{1}{2}(1-\sqrt{\cmax })\log_2 \bigg(\frac{ \cmax }{c_2} \bigg).
\end{equation}
\end{corollary}

Eq. \eqref{eq:lowerboundapp}, or more precisely its special case~\eqref{eqn23779af} for complete von Neumann measurements, allows us to argue that, if $d\geq 3$, our bound $\qq '$ (and hence also $q$) is an improvement over the standard bound $\qmu $ for almost all pairs of orthonormal bases $( X,Z )$. To argue this, we will need the following lemma, proved in App.~\ref{app:prooflem}, kindly provided by N. Johnston \cite{NJohnstonPrivComm}.
\begin{lemma}
\label{lem:generic}
For any dimension $d\geq 3$, the entries $U_{ij}$ of a generic $d$-dimensional unitary $U$ satisfy
\begin{equation}
\label{eq:typicalunitary}
|U_{ij}|\neq |U_{kl}|,\quad \forall (i,j)\neq (k,l).
\end{equation}
That is, for any dimension $d\geq 3$ the set of unitaries that violate \eqref{eq:typicalunitary} has vanishing measure with respect to the Haar measure.
\end{lemma}
Combining Corollary~\ref{thmAppQprimebound} with Lemma~\ref{lem:generic} immediately leads to the following conclusion.
\begin{corollary}
In any dimension $d\geq 3$, for almost all choices of two orthonormal bases one has $c_2< \cmax  <1 $, hence $\qq ' > \qmu $.
\end{corollary}
\begin{proof}
The bases are related by a unitary transformation $U$, represented in the first basis by entries $U_{ij}$. The parameter $\cmax $ corresponds to the modulus square of the largest entry. Because of Lemma~\ref{lem:generic}, we have $\cmax <1$ generically. Indeed, if $\cmax =1$, all the remaining entries in the same row or column must vanish, violating \eqref{eq:typicalunitary}. Finally, also the condition $c_2=\cmax $ corresponds to a violation of \eqref{eq:typicalunitary}. The use of \eqref{eq:lowerboundapp} completes the claim.
\end{proof}

\subsection{Arbitrarily large gap between $\qmu $ and our bound }\label{apsct2.3}

We show here that the gap between $\qmu $ and our state-independent bound can grow unboundedly, and more precisely logarithmically in the dimension of the system involved. In the case of $q$, this can be seen by observing that it is additive on tensor copies, i.e., $\qq (\{X^{\otimes n},Z^{\otimes n}\}) = n \qq (\{X,Z\})$ where $X$ and $Z$ are arbitrary POVMs. Since $\qmu $ is also additive on tensor copies, any gap between $\qmu $ and $q$ for a single copy of $X$ and $Z$ will get multiplied by $n$.

Our simple analytical bound $q'$ is not additive on tensor copies. Nonetheless, we construct the following example for which $\dl := q' - \qmu$ grows as $\log_2 d$. Consider a Hilbert space $\HC_A = \HC_{A_1}\oplus \HC_{A_2}  $ with $\dim (\HC_A) = d$, $\dim (\HC_{A_1}) = 1$, $\dim (\HC_{A_2}) = d-1$. Let
\begin{equation}
U_0 =\id_{1} \oplus F_{d-1}
\end{equation}
be a unitary acting on $\HC_A$ where $\id_{1} $ is the $1\times 1$ identity matrix (acting on $\HC_{A_1}$). Also, $F_{d-1} = \sum_{j,k} \frac{\omega^{jk}}{\sqrt{d-1}} \dyad{j}{k}=\sum_{j} \dyad{t_j}{j}$, with $\omega = e^{2\pi i/(d-1)}$, is the Fourier matrix of dimension $d-1$, which acts on $\HC_{A_2}$ by mapping the standard basis $S= \{\ket{j}\}$ to the basis $T = \{\ket{t_j}\}$. We suppose that the orthonormal bases on $\HC_A$ of interest ($X$ and $Z$) for the uncertainty relation are related by a unitary $U$ that is the product of $U_0$ and a slight rotation $U_r$, i.e., 
\begin{equation}
U = U_r U_0.
\end{equation}
Now let $\ket{y_0}\in \HC_{A_2}$ be a state that is unbiased with respect to both the $S$ basis and the $T$ basis on $\HC_{A_2}$. (It is always possible to find such a state regardless of the Hilbert space dimension.) We define $U_r$ by
\begin{equation}
 U_r = e^{-i H_r \th} , \quad H_r = \dyad{y'_0}{0}+\dyad{0}{y'_0},
 \end{equation}
where $\ket{y'_0} = V \ket{y_0}$ and $V:\HC_{A_2} \to \HC_{A}$ is an isometry that embeds the $d-1$ dimensional space $\HC_{A_2}$ into the $d$ dimensional space $\HC_{A}$ defined by $V = \sum_{j=0}^{d-2} \dyad{j+1}{j}$. We choose the rotation angle $0 < \th < \pi / 2$ to be a constant, i.e., independent of $d$. Note that
\begin{align}
H_r^2 = \dya{0}+ \dya{y'_0}, \quad H_r^3 = H_r, \quad H_r^4 = H_r^2, \notag
\end{align}
which implies that 
\begin{align}
\sin (H_r \th )&= H_r \cdot \sin \th, \notag\\
\cos (H_r \th)&= (\id-H_r^2)+ H_r^2  \cdot \cos \th, \notag
\end{align}
and 
\begin{align}
U_r &= \cos (H_r \th) - i \sin (H_r \th) \notag\\
&= (\id-H_r^2)+ H_r^2  \cdot \cos \th - i H_r \cdot \sin \th.
\end{align}
For $j\neq 0$ and $k\neq 0$, we have 
\begin{align}
\mte{0}{U} &= \cos (\th),\notag\\
\mted{0}{U}{j} &= -i \sin (\th) \ip{y'_0}{j}, \notag\\
\mted{j}{U}{0}& = -i \sin (\th) \ip{j}{y'_0}, \notag\\
\mted{j}{U}{k} &=  \mted{j}{F'_{d-1}}{k}  +(\cos(\th)-1) \ip{j}{y'_0} \mted{y'_0}{F'_{d-1}}{k}, \notag
\end{align}
where we write $F'_{d-1} = V F_{d-1} V\ad$ for clarity.
In the limit of large $d$, this gives
\begin{align}
c_{00} &= |\mte{0}{U}|^2 = \cos^2 \th,\notag\\
c_{0j} &=| \mted{0}{U}{j} |^2 \approx (1/d) \sin^2 \th , \notag\\
c_{j0}&=| \mted{j}{U}{0} |^2 \approx (1/d) \sin^2 \th , \notag\\
c_{jk} & = | \mted{j}{U}{k}  |^2 \approx  1/d . 
\end{align}
Thus, in this limit, we have
\begin{equation}
\cmax = \cos^2 \th , \quad c_2 \approx 1/d.
\end{equation}
So for large $d$ the gap is given by
\begin{align} 
\dl &:= q' - \qmu \notag\\
&= \frac{1}{2}(1-\sqrt{\cmax })\log_2\frac{\cmax }{c_2}\notag\\
&\approx \frac{1}{2}(1-\cos \th )\log_2 (d \cos^2 \th ).
\end{align}
So $\dl$ grows with $\log_2 d$ in this example.

\section{Conclusions}

We gave two main results: we strengthened the bound in the uncertainty principle with quantum memory, and we formulated an information exclusion relation (a bound on complementary mutual information terms) that also allows for quantum memory. The latter is a major improvement over previously known information exclusion relations, with a much stronger bound that even provides qualitatively new insight into the complementarity of information and how it differs from that of uncertainty. Our results have applications in, e.g., quantum cryptography, entanglement verification and quantum communication. It would be interesting to see if our results extend to smooth entropies or smooth mutual informations that are relevant to non-asymptotic information theory \cite{RennerThesis05URL, TomamichelThesis2012}.

\section{Acknowledgments}

We thank Marco Tomamichel, Koenraad Audenaert and Maris Ozols for helpful discussions. We thank Nathaniel Johnston for providing the proof of Lemma~\ref{lem:generic}. PJC is funded by the Ministry of Education (MOE) and National Research Foundation Singapore, as well as MOE Tier 3 Grant ``Random numbers from quantum processes" (MOE2012-T3-1-009). MP acknowledges support from NSERC, CIFAR, DARPA, and Ontario Centres of Excellence.

\appendix

\section{Proofs of Technical Results}

\subsection{Proof of Thm.~\ref{thmMain230}}
\label{app:proofthm}

Theorem~\ref{thmMain230} is a particular case of Theorem~\ref{thmUncRelApp}, whose proof we provide in Appendix~\ref{app:proofthm2}, but it is instructive to derive it directly.

\begin{proof}
We use basic properties of the relative entropy $D(\rho || \sg) := \Tr(\rho\log_2\rho)- \Tr(\rho\log_2\sg)$ \cite{ColesColbeckYuZwolak2012PRL}. We first start with an identity \cite{ColesEtAl,ColesDecDisc2012} that relates conditional entropy to relative entropy and proceed as follows:
\begin{align}
H&(Z|B) - H(A|B) = D(\rho_{AB}|| \sum_k \dya{z_k}\rho_{AB}\dya{z_k})\notag \\ 
&\geq D(\rho_{XB}|| \sum_{j,k} c_{jk} \dya{x_j} \ot \Tr_A(\dya{z_k} \rho_{AB}) )\\
&\geq D(\rho_{XB}|| \sum_{j} \max_k(c_{jk}) \dya{x_j} \ot \rho_B )\\
&= - H(X|B) -\Tr(\rho_X \log_2 \sum_{j} \max_k(c_{jk}) \dya{x_j} )\\
&= - H(X|B)+ q(\rho_A, X,Z),
\end{align}
where the second line used the monotonicity of relative entropy under quantum channels, for the channel $\XC$. The third line used the fact that $D(S || T) \geq D(S || T')$ if $T' \geq T$, and the fourth line used the definition of $D(\rho ||\sg)$. By symmetry the same proof works by interchanging $X$ and $Z$, thus either $q(\rho_A, X,Z)$ or $q(\rho_A, Z,X)$ can be used in the bound, so we take the maximum as in $q(\rho_A) $.
\end{proof}

\subsection{Proof of Cor.~\ref{thmMain230a}}
\label{app:proofcor}

Corollary~\ref{thmMain230a} is a particular case of Corollary~\ref{thmAppQprimebound}, whose proof we provide in Appendix~\ref{app:proofcor2}, but it is instructive to derive it directly.

\begin{proof}
Let $(\hat{j},\hat{k})$ be such that $c_{\hat{j}\hat{k}}=\cmax$, and $\lm_{\max}(\cdot)$ indicates the maximum eigenvalue. Then, from Theorem~\ref{thmMain230} and definitions~\eqref{eq:qprime},
\begin{align}
q(\rho_A)&\geq \frac{1}{2}[- \sum_j p^x_j \log_2 (\max_k c_{jk} )- \sum_k p^z_k \log_2 (\max_j c_{jk} )] \notag\\
&\geq \frac{1}{2} [- p^x_{\hat{j}} \log_2 \cmax - (1-p^x_{\hat{j}}) \log_2 c_2\notag\\
&\quad - p^z_{\hat{k}} \log_2 \cmax - (1-p^z_{\hat{k}}) \log_2 c_2 ]\notag \\
&=q_{\textup{MU}} +\frac{1}{2}\log_2 \bigg(\frac{\cmax}{c_2}\bigg) [2 - ( p_{\hat{j}}^x +  p_{\hat{k}}^z) ]\notag\\
&\geq q_{\textup{MU}} +\frac{1}{2}\log_2 \bigg(\frac{\cmax}{c_2}\bigg) [2 - \lm_{\max}( \dya{x_{\hat{j}}} + \dya{z_{\hat{k}}}) ].\notag
\end{align}
The last line is equal to $q'$ in~\eqref{eqn23779af}, completing the proof, since $\lm_{\max}( \dya{x_{\hat{j}}} + \dya{z_{\hat{k}}}) = 1+\sqrt{\cmax}$.
\end{proof}

\subsection{Proof of Thm.~\ref{thmUncRelApp}}
\label{app:proofthm2}

Let us first state the following lemma that is used in proving the uncertainty relation. The lemma was given in \cite{ColesEtAl}, but we reproduce its proof here for completeness.
\begin{lemma}\label{thm3875} \cite{ColesEtAl}
Let $Z = \{Z_k\}$ be any POVM on system $A$, then for any tripartite state $\rho_{ABC}$,
\begin{equation}
H(Z|C) \geq D(\rho_{AB} || \sum_k Z_k \rho_{AB} Z_k).
\end{equation}
\end{lemma}
\begin{proof}
Consider the state $\rhob_{ZZ'ABC}$ defined in \eqref{eqnPM11}. Applying strong subadditivity to this state gives $H(Z|C)+H(Z|Z'AB)\geq 0$. Now note that conditional entropy can be rewritten in terms of relative entropy with the formula $-H(A|B)_{\sg} = D(\sg_{AB}||\id \ot \sg_B)$. So we have:
\begin{align}
H(Z|C) &\geq  - H(Z|Z'AB)\\
&= D(\rhob_{ZZ'AB}|| \id \ot \rhob_{Z'AB})\\
&\geq  D(\rhob_{ZZ'AB}|| V_ZV_Z\ad (\id \ot \rhob_{Z'AB})V_ZV_Z\ad)\\
&=  D(\rho_{AB}|| V_Z\ad (\id \ot \rhob_{Z'AB})V_Z)\\
&=  D(\rho_{AB}|| \sum_k Z_k \rho_{AB} Z_k).
\end{align}
The third line used the property $D(\rho || \sg)\geq D(\rho || \Pi_{\rho} \sg\Pi_{\rho})$ where $\Pi_{\rho}$ is a projector onto a space that includes the support of $\rho$; in this case we chose $\Pi_{\rho} = V_ZV_Z\ad$. The fourth line used the invariance of relative entropy under isometries. It is straightforward to verify the fifth line using $\rhob_{Z'AB} = \sum_k \dya{k} \ot \sqrt{Z_k}\rho_{AB}\sqrt{Z_k}$.
\end{proof}

Now we prove Thm.~\ref{thmUncRelApp}.

\begin{proof}
Starting from Lemma~\ref{thm3875} we invoke the data-processing inequality for the quantum channel $\XC$ in \eqref{eqnXchannel11}, as follows
\begin{align}
&H(Z|C) \geq D(\rho_{AB} || \sum_k Z_k \rho_{AB} Z_k)\notag\\
&\geq D(\rhoh_{XB} || \sum_{j,k} \dya{j} \ot \Tr_A(Z_k X_j Z_k \rho_{AB}))\\
&\geq D(\rhoh_{XB} || \sum_j  h_j(X,Z) \dya{j}\ot  \rho_{B} )\\
&=  - H(\rhoh_{XB}) - \Tr_{XB} [ \rhoh_{XB} \log_2  \sum_j  h_j(X,Z) \dya{j}\ot  \rho_{B} ] \\
&=  - H(X|B) - \Tr_X [\rho_{X} \log_2  \sum_j  h_j(X,Z) \dya{j} ]\\
&=  - H(X|B) + \qq (\rho_A, X,Z)
\end{align}
where the fifth line used the additivity of the log for tensor products. The third line invoked the property $D(S||T)\geq D(S||T')$ if $T' \geq T$, where we note that\begin{align}
\sum_{j,k} &\dya{j}\ot \Tr_A(Z_k X_j Z_k \rho_{AB}) \notag\\
&= \sum_{j}  \dya{j}\ot \Tr_A[(\sum_k Z_k X_j Z_k) \rho_{AB}] \notag\\
&\leq \sum_j  h_j(X,Z) \dya{j}\ot  \rho_{B} 
\end{align}
since $\sum_k Z_k X_j Z_k \leq \| \sum_k Z_k X_j Z_k \|_{\infty}\id$.

Finally, by symmetry, one can interchange $X$ and $Z$ in the bound and hence use $\qq (\rho_A)$.
\end{proof}

\subsection{Proof of Lem.~\ref{thmPOVM3875}}
\label{app:prooflemA}

\begin{proof}
First notice that
$$\max_k \big\| \sqrt{Z_k} \sg \sqrt{Z_k} \big\|_{\infty} =  \big\| \rho \big\|_{\infty},$$
where
$$\rho:= \sum_k \dya{k}\ot \sqrt{Z_k} \sg \sqrt{Z_k}$$
and $\{\ket{k}\}$ is the standard basis on an auxiliary space. Now consider the isometry $V = \sum_k \ket{k} \ot \sqrt{Z_k}$ and notice that
$$\sum_k Z_k \sg Z_k=V\ad \rho V.$$
So we wish to show that
$$\big\|  \rho \big\|_{\infty}\geq \big\|  V\ad \rho V \big\|_{\infty}.$$
Consider the projector $\Pi = VV\ad$ and the channel $\EC(\cdot) = \Pi(\cdot) \Pi + (\id- \Pi)(\cdot) (\id- \Pi)$ that pinches with respect to this projector. It is a standard result in matrix analysis that the infinity norm never increases upon pinching the argument~\cite{bhatia1997matrix}. So we have
\begin{align}
\| \rho \|_{\infty} &\geq \| \EC(\rho) \|_{\infty} \notag\\
&= \max\{ \| \Pi \rho \Pi \|_{\infty}, \| (\id - \Pi )\rho (\id - \Pi ) \|_{\infty}\} \notag\\
&\geq \| \Pi \rho \Pi \|_{\infty} = \| V\ad \rho V \|_{\infty},
\end{align}
where the last equality uses the invariance of the norm under isometries. 
\end{proof}

\subsection{Proof of Cor.~\ref{thmAppQprimebound}}
\label{app:proofcor2}

We first note the following useful lemma, shown, e.g., in Refs.~\cite{TomEtAl2012arXiv1210.4359T, SchaffnerThesis2007}.
\begin{lemma}
\label{thmPosOp1}
For any positive semi-definite operators $S\geq~0$ and $T\geq 0$, we have
\begin{equation}
\label{eqposoplemm1}
\| S + T \|_{\infty} \leq \max \{ \| S  \|_{\infty},\|  T \|_{\infty}\} + \| \sqrt{S} \sqrt{T} \|_{\infty}.
\end{equation}
\end{lemma}

Now we prove Cor.~\ref{thmAppQprimebound}.
\begin{proof}
From the definition \eqref{eqnqprimerhoaApp} of $\qq (\rho_A)$, and using \eqref{eqnhj5}, we have
\begin{align}
\qq  (\rho_A)&\geq  \max\{- \sum_j p^x_j \log_2 (\max_k c_{jk} ), -\sum_k p^z_k \log_2 (\max_j c_{jk} )\}\notag\\
&\geq \frac{1}{2}[- \sum_j p^x_j \log_2 (\max_k c_{jk} )- \sum_k p^z_k \log_2 (\max_j c_{jk} )] \notag\\
&\geq \frac{1}{2} [- p^x_{\hat{j}} \log_2 \cmax - (1-p^x_{\hat{j}}) \log_2 c_2\notag\\
&\quad - p^z_{\hat{k}} \log_2 \cmax - (1-p^z_{\hat{k}}) \log_2 c_2 ]\notag \\
\label{eqnApp3242}&=\qmu +\frac{1}{2}\log_2 \bigg(\frac{\cmax }{c_2}\bigg) [2 - ( p_{\hat{j}}^x +  p_{\hat{k}}^z) ].
\end{align}
Since $\log_2 (\cmax /c_2) \geq 0$ by assumption, to bound $\qq  = \min_{\rho_A} \qq (\rho_A)$ we need to evaluate
\begin{align}
\max_{\rho_A} (p_{\hat{j}}^x +  p_{\hat{k}}^z)&=\max_{\rho_A} \Tr [\rho_A ( X_{\hat{j}} + Z_{\hat{k}})]\notag\\
&= \| X_{\hat{j}} + Z_{\hat{k}} \|_\infty \notag\\
&\leq \max\{\|X_{\hat{j}}\|_\infty,\|Z_{\hat{k}}\|_\infty\}+  \| \sqrt{X_{\hat{j}}}\sqrt{Z_{\hat{k}}}  \|_\infty \notag\\
&\leq 1+ \sqrt{\cmax },
\end{align}
where in the first inequality we have used \eqref{eqposoplemm1} from Lemma \ref{thmPosOp1}, and in the second inequality the fact that $X_{\hat{j}}$ and $Z_{\hat{k}}$, being POVM elements,  both have operator norm less than unity. Plugging this into \eqref{eqnApp3242} proves~\eqref{eq:lowerboundapp}.
\end{proof}

\subsection{Proof of Lem.~\ref{lem:generic}}
\label{app:prooflem}

\begin{proof}
The proof relies on concepts of algebraic geometry~\cite{humphreys1975linear}. The set of unitaries in dimension $d$ has real dimension $d^2$, that is, one has to specify $d^2$ real parameters to specify a unitary $U$. On the other hand, unitaries can be seen as forming a real algebraic variety $\mathcal{U}$ in $\mathbb{R}^{2d^2}$. Indeed, let the real numbers $x_{kl}$ and $y_{kl}$ be the real and imaginary components of the matrix entry $U_{kl}$, i.e., $U_{kl}= x_{kl} + i y_{kl}$. Then the condition $U^\dagger U=\openone$ corresponds to a system of quadratic equations in the $x_{kl}$'s and $y_{kl}$'s. Since the unitaries form a connected group, the algebraic variety $\mathcal{U}$ is irreducible~\cite{humphreys1975linear}. In particular, if $\mathcal{Z}$ is another algebraic variety, either $\mathcal{U}\cap \mathcal{Z}$ is equal to $\mathcal{U}$ (if $\mathcal{U}\subseteq\mathcal{Z}$) or $\mathcal{U}\cap \mathcal{Z}$ has real dimension strictly smaller than $d^2$. For any choice of two ordered pairs $(i,j)$ and $(k,l)$, consider the algebraic variety $\mathcal{Z}_{(i,j)(k,l)}\subseteq\mathbb{R}^{2d^2}$ defined by $(x^2_{ij}+y_{ij}^2)-(x^2_{kl}+y_{kl}^2) = 0$. It is easy to check that for $d\geq 3$, $\mathcal{U}\not\subseteq \mathcal{Z}_{(i,j)(k,l)}$ for every choice of $(i,j)\neq (k,l)$. This is because, when $d\geq 3$, for any $(i,j)\neq (k,l)$ it is possible to find a unitary that does not belong to $\mathcal{Z}_{(i,j)(k,l)}$. Notice that, on the other hand, $\mathcal{U}=\mathcal{Z}_{(1,1)(2,2)}=\mathcal{Z}_{(1,2)(2,1)}$ for $d=2$. Thus, for $d\geq 3$ and $(i,j)\neq (k,l)$, $\mathcal{U}\cap \mathcal{Z}_{(i,j)(k,l)}$ has real dimension strictly less than $d^2$, hence vanishing Haar measure. Given that there is a finite number of sets $Z_{(i,j)(k,l)}$, it also holds that the union of all $\mathcal{U}\cap\mathcal{Z}_{(i,j)(k,l)}$'s has real dimension strictly less than $d^2$ and vanishing Haar measure. The claim follows.
\end{proof}


\begin{thebibliography}{37}
\expandafter\ifx\csname natexlab\endcsname\relax\def\natexlab#1{#1}\fi
\expandafter\ifx\csname bibnamefont\endcsname\relax
  \def\bibnamefont#1{#1}\fi
\expandafter\ifx\csname bibfnamefont\endcsname\relax
  \def\bibfnamefont#1{#1}\fi
\expandafter\ifx\csname citenamefont\endcsname\relax
  \def\citenamefont#1{#1}\fi
\expandafter\ifx\csname url\endcsname\relax
  \def\url#1{\texttt{#1}}\fi
\expandafter\ifx\csname urlprefix\endcsname\relax\def\urlprefix{URL }\fi
\providecommand{\bibinfo}[2]{#2}
\providecommand{\eprint}[2][]{\url{#2}}

\bibitem[{\citenamefont{Heisenberg}(1927)}]{Heisenberg}
\bibinfo{author}{\bibfnamefont{W.}~\bibnamefont{Heisenberg}},
  \bibinfo{journal}{Zeitschrift f\"ur Physik} \textbf{\bibinfo{volume}{43}},
  \bibinfo{pages}{172} (\bibinfo{year}{1927}).

\bibitem[{\citenamefont{Kennard}(1927)}]{kennard1927quantum}
\bibinfo{author}{\bibfnamefont{E.}~\bibnamefont{Kennard}}, \bibinfo{journal}{Z.
  Phys} \textbf{\bibinfo{volume}{44}}, \bibinfo{pages}{326}
  (\bibinfo{year}{1927}).

\bibitem[{\citenamefont{Robertson}(1929)}]{Robertson}
\bibinfo{author}{\bibfnamefont{H.~P.} \bibnamefont{Robertson}},
  \bibinfo{journal}{Phys. Rev.} \textbf{\bibinfo{volume}{34}},
  \bibinfo{pages}{163} (\bibinfo{year}{1929}).

\bibitem[{\citenamefont{Maassen and Uffink}(1988)}]{MaassenUffink}
\bibinfo{author}{\bibfnamefont{H.}~\bibnamefont{Maassen}} \bibnamefont{and}
  \bibinfo{author}{\bibfnamefont{J.~B.~M.} \bibnamefont{Uffink}},
  \bibinfo{journal}{Phys. Rev. Lett.} \textbf{\bibinfo{volume}{60}},
  \bibinfo{pages}{1103} (\bibinfo{year}{1988}).

\bibitem[{\citenamefont{Wiesner}(1983)}]{Wiesner}
\bibinfo{author}{\bibfnamefont{S.}~\bibnamefont{Wiesner}},
  \bibinfo{journal}{SIGACT News} \textbf{\bibinfo{volume}{15}},
  \bibinfo{pages}{78} (\bibinfo{year}{1983}).

\bibitem[{\citenamefont{Horodecki et~al.}(2009)\citenamefont{Horodecki,
  Horodecki, Horodecki, and Horodecki}}]{HHHH09}
\bibinfo{author}{\bibfnamefont{R.}~\bibnamefont{Horodecki}},
  \bibinfo{author}{\bibfnamefont{P.}~\bibnamefont{Horodecki}},
  \bibinfo{author}{\bibfnamefont{M.}~\bibnamefont{Horodecki}},
  \bibnamefont{and}
  \bibinfo{author}{\bibfnamefont{K.}~\bibnamefont{Horodecki}},
  \bibinfo{journal}{Rev. Mod. Phys.} \textbf{\bibinfo{volume}{81}},
  \bibinfo{pages}{865} (\bibinfo{year}{2009}).

\bibitem[{\citenamefont{{Berta} et~al.}(2010)\citenamefont{{Berta},
  {Christandl}, {Colbeck}, {Renes}, and {Renner}}}]{BertaEtAl}
\bibinfo{author}{\bibfnamefont{M.}~\bibnamefont{{Berta}}},
  \bibinfo{author}{\bibfnamefont{M.}~\bibnamefont{{Christandl}}},
  \bibinfo{author}{\bibfnamefont{R.}~\bibnamefont{{Colbeck}}},
  \bibinfo{author}{\bibfnamefont{J.~M.} \bibnamefont{{Renes}}},
  \bibnamefont{and} \bibinfo{author}{\bibfnamefont{R.}~\bibnamefont{{Renner}}},
  \bibinfo{journal}{Nature Physics} \textbf{\bibinfo{volume}{6}},
  \bibinfo{pages}{659} (\bibinfo{year}{2010}).

\bibitem[{\citenamefont{Prevedel et~al.}(2011)\citenamefont{Prevedel, Hamel,
  Colbeck, Fisher, and Resch}}]{PHCFR}
\bibinfo{author}{\bibfnamefont{R.}~\bibnamefont{Prevedel}},
  \bibinfo{author}{\bibfnamefont{D.~R.} \bibnamefont{Hamel}},
  \bibinfo{author}{\bibfnamefont{R.}~\bibnamefont{Colbeck}},
  \bibinfo{author}{\bibfnamefont{K.}~\bibnamefont{Fisher}}, \bibnamefont{and}
  \bibinfo{author}{\bibfnamefont{K.~J.} \bibnamefont{Resch}},
  \bibinfo{journal}{Nature Physics} \textbf{\bibinfo{volume}{7}},
  \bibinfo{pages}{757} (\bibinfo{year}{2011}).

\bibitem[{\citenamefont{Li et~al.}(2011)\citenamefont{Li, Xu, Xu, Li, and
  Guo}}]{LXXLG}
\bibinfo{author}{\bibfnamefont{C.-F.} \bibnamefont{Li}},
  \bibinfo{author}{\bibfnamefont{J.-S.} \bibnamefont{Xu}},
  \bibinfo{author}{\bibfnamefont{X.-Y.} \bibnamefont{Xu}},
  \bibinfo{author}{\bibfnamefont{K.}~\bibnamefont{Li}}, \bibnamefont{and}
  \bibinfo{author}{\bibfnamefont{G.-C.} \bibnamefont{Guo}},
  \bibinfo{journal}{Nature Physics} \textbf{\bibinfo{volume}{7}},
  \bibinfo{pages}{752} (\bibinfo{year}{2011}).

\bibitem[{\citenamefont{{Tomamichel} and {Renner}}(2011)}]{TomRen2010}
\bibinfo{author}{\bibfnamefont{M.}~\bibnamefont{{Tomamichel}}}
  \bibnamefont{and} \bibinfo{author}{\bibfnamefont{R.}~\bibnamefont{{Renner}}},
  \bibinfo{journal}{Phys. Rev. Lett.} \textbf{\bibinfo{volume}{106}},
  \bibinfo{pages}{110506} (\bibinfo{year}{2011}).

\bibitem[{\citenamefont{Tomamichel et~al.}(2012)\citenamefont{Tomamichel, Lim,
  Gisin, and Renner}}]{TLGR}
\bibinfo{author}{\bibfnamefont{M.}~\bibnamefont{Tomamichel}},
  \bibinfo{author}{\bibfnamefont{C.~C.~W.} \bibnamefont{Lim}},
  \bibinfo{author}{\bibfnamefont{N.}~\bibnamefont{Gisin}}, \bibnamefont{and}
  \bibinfo{author}{\bibfnamefont{R.}~\bibnamefont{Renner}},
  \bibinfo{journal}{Nature Communications} \textbf{\bibinfo{volume}{3}},
  \bibinfo{pages}{634} (\bibinfo{year}{2012}).

\bibitem[{\citenamefont{Hall}(1995)}]{Hall1}
\bibinfo{author}{\bibfnamefont{M.~J.~W.} \bibnamefont{Hall}},
  \bibinfo{journal}{Phys. Rev. Lett.} \textbf{\bibinfo{volume}{74}},
  \bibinfo{pages}{3307} (\bibinfo{year}{1995}).

\bibitem[{\citenamefont{Cover and Thomas}(2005)}]{CvTh06}
\bibinfo{author}{\bibfnamefont{T.~M.} \bibnamefont{Cover}} \bibnamefont{and}
  \bibinfo{author}{\bibfnamefont{J.~A.} \bibnamefont{Thomas}},
  \emph{\bibinfo{title}{Elements of Information Theory}}
  (\bibinfo{publisher}{Wiley}, \bibinfo{address}{New York},
  \bibinfo{year}{2005}), \bibinfo{edition}{2nd} ed.

\bibitem[{\citenamefont{Hall}(1997)}]{HallPRA1997}
\bibinfo{author}{\bibfnamefont{M.~J.~W.} \bibnamefont{Hall}},
  \bibinfo{journal}{Phys. Rev. A} \textbf{\bibinfo{volume}{55}},
  \bibinfo{pages}{100} (\bibinfo{year}{1997}).

\bibitem[{\citenamefont{Coles et~al.}(2011)\citenamefont{Coles, Yu, Gheorghiu,
  and Griffiths}}]{ColesEtAl}
\bibinfo{author}{\bibfnamefont{P.~J.} \bibnamefont{Coles}},
  \bibinfo{author}{\bibfnamefont{L.}~\bibnamefont{Yu}},
  \bibinfo{author}{\bibfnamefont{V.}~\bibnamefont{Gheorghiu}},
  \bibnamefont{and} \bibinfo{author}{\bibfnamefont{R.~B.}
  \bibnamefont{Griffiths}}, \bibinfo{journal}{Phys. Rev. A}
  \textbf{\bibinfo{volume}{83}}, \bibinfo{pages}{062338}
  (\bibinfo{year}{2011}).

\bibitem[{\citenamefont{{Grudka} et~al.}(2013)\citenamefont{{Grudka},
  {Horodecki}, {Horodecki}, {Horodecki}, {K{\l}obus}, and
  {Pankowski}}}]{GrudkaEtAl2012}
\bibinfo{author}{\bibfnamefont{A.}~\bibnamefont{{Grudka}}},
  \bibinfo{author}{\bibfnamefont{M.}~\bibnamefont{{Horodecki}}},
  \bibinfo{author}{\bibfnamefont{P.}~\bibnamefont{{Horodecki}}},
  \bibinfo{author}{\bibfnamefont{R.}~\bibnamefont{{Horodecki}}},
  \bibinfo{author}{\bibfnamefont{W.}~\bibnamefont{{K{\l}obus}}},
  \bibnamefont{and}
  \bibinfo{author}{\bibfnamefont{{\L}.}~\bibnamefont{{Pankowski}}},
  \bibinfo{journal}{Phys. Rev. A} \textbf{\bibinfo{volume}{88}},
  \bibinfo{pages}{032106} (\bibinfo{year}{2013}), \eprint{1210.8317}.

\bibitem[{\citenamefont{Pati et~al.}(2012)\citenamefont{Pati, Wilde, Devi,
  Rajagopal, and Sudha}}]{PatiEtAlPRA2012}
\bibinfo{author}{\bibfnamefont{A.~K.} \bibnamefont{Pati}},
  \bibinfo{author}{\bibfnamefont{M.~M.} \bibnamefont{Wilde}},
  \bibinfo{author}{\bibfnamefont{A.~R.~U.} \bibnamefont{Devi}},
  \bibinfo{author}{\bibfnamefont{A.~K.} \bibnamefont{Rajagopal}},
  \bibnamefont{and} \bibinfo{author}{\bibnamefont{Sudha}},
  \bibinfo{journal}{Phys. Rev. A} \textbf{\bibinfo{volume}{86}},
  \bibinfo{pages}{042105} (\bibinfo{year}{2012}).

\bibitem[{\citenamefont{{Tomamichel}}(2012)}]{TomamichelThesis2012}
\bibinfo{author}{\bibfnamefont{M.}~\bibnamefont{{Tomamichel}}}, Ph.D. thesis,
  \bibinfo{school}{ETH Z\"{u}rich} (\bibinfo{year}{2012}),
  \urlprefix\url{http://arxiv.org/abs/1203.2142}.

\bibitem[{\citenamefont{Tomamichel and H{\"a}nggi}(2013)}]{TomHang2013}
\bibinfo{author}{\bibfnamefont{M.}~\bibnamefont{Tomamichel}} \bibnamefont{and}
  \bibinfo{author}{\bibfnamefont{E.}~\bibnamefont{H{\"a}nggi}},
  \bibinfo{journal}{Journal of Physics A: Mathematical and Theoretical}
  \textbf{\bibinfo{volume}{46}}, \bibinfo{pages}{055301}
  (\bibinfo{year}{2013}).

\bibitem[{\citenamefont{de~Vicente and
  S\'anchez-Ruiz}(2008)}]{deVicSanRuizPRA2008}
\bibinfo{author}{\bibfnamefont{J.~I.} \bibnamefont{de~Vicente}}
  \bibnamefont{and}
  \bibinfo{author}{\bibfnamefont{J.}~\bibnamefont{S\'anchez-Ruiz}},
  \bibinfo{journal}{Phys. Rev. A} \textbf{\bibinfo{volume}{77}},
  \bibinfo{pages}{042110} (\bibinfo{year}{2008}).

\bibitem[{\citenamefont{{Friedland} et~al.}(2013)\citenamefont{{Friedland},
  {Gheorghiu}, and {Gour}}}]{FriedGheorGour2013arxiv}
\bibinfo{author}{\bibfnamefont{S.}~\bibnamefont{{Friedland}}},
  \bibinfo{author}{\bibfnamefont{V.}~\bibnamefont{{Gheorghiu}}},
  \bibnamefont{and} \bibinfo{author}{\bibfnamefont{G.}~\bibnamefont{{Gour}}},
  \bibinfo{journal}{ArXiv e-prints}  (\bibinfo{year}{2013}),
  \eprint{1304.6351}.

\bibitem[{\citenamefont{Pucha{\l}a et~al.}(2013)\citenamefont{Pucha{\l}a,
  Rudnicki, and {\.Z}yczkowski}}]{PuchalaEtAl2013}
\bibinfo{author}{\bibfnamefont{Z.}~\bibnamefont{Pucha{\l}a}},
  \bibinfo{author}{\bibfnamefont{{\L}.}~\bibnamefont{Rudnicki}},
  \bibnamefont{and}
  \bibinfo{author}{\bibfnamefont{K.}~\bibnamefont{{\.Z}yczkowski}},
  \bibinfo{journal}{Journal of Physics A: Mathematical and Theoretical}
  \textbf{\bibinfo{volume}{46}}, \bibinfo{pages}{272002}
  (\bibinfo{year}{2013}).

\bibitem[{\citenamefont{Modi et~al.}(2012)\citenamefont{Modi, Brodutch, Cable,
  Paterek, and Vedral}}]{ModiEtAlRevModPhys.84.1655}
\bibinfo{author}{\bibfnamefont{K.}~\bibnamefont{Modi}},
  \bibinfo{author}{\bibfnamefont{A.}~\bibnamefont{Brodutch}},
  \bibinfo{author}{\bibfnamefont{H.}~\bibnamefont{Cable}},
  \bibinfo{author}{\bibfnamefont{T.}~\bibnamefont{Paterek}}, \bibnamefont{and}
  \bibinfo{author}{\bibfnamefont{V.}~\bibnamefont{Vedral}},
  \bibinfo{journal}{Rev. Mod. Phys.} \textbf{\bibinfo{volume}{84}},
  \bibinfo{pages}{1655} (\bibinfo{year}{2012}).

\bibitem[{\citenamefont{Christandl and Winter}(2005)}]{ChristWinterIEEE2005}
\bibinfo{author}{\bibfnamefont{M.}~\bibnamefont{Christandl}} \bibnamefont{and}
  \bibinfo{author}{\bibfnamefont{A.}~\bibnamefont{Winter}},
  \bibinfo{journal}{IEEE Trans. Inf. Theory} \textbf{\bibinfo{volume}{51}},
  \bibinfo{pages}{3159} (\bibinfo{year}{2005}).

\bibitem[{\citenamefont{Lloyd}(1997)}]{Lloyd97}
\bibinfo{author}{\bibfnamefont{S.}~\bibnamefont{Lloyd}},
  \bibinfo{journal}{Phys. Rev. A} \textbf{\bibinfo{volume}{55}},
  \bibinfo{pages}{1613} (\bibinfo{year}{1997}).

\bibitem[{\citenamefont{Renes and Boileau}(2009)}]{RenesBoileau}
\bibinfo{author}{\bibfnamefont{J.~M.} \bibnamefont{Renes}} \bibnamefont{and}
  \bibinfo{author}{\bibfnamefont{J.-C.} \bibnamefont{Boileau}},
  \bibinfo{journal}{Phys. Rev. Lett.} \textbf{\bibinfo{volume}{103}},
  \bibinfo{pages}{020402} (\bibinfo{year}{2009}).

\bibitem[{\citenamefont{Renner}(2005)}]{RennerThesis05URL}
\bibinfo{author}{\bibfnamefont{R.}~\bibnamefont{Renner}}, Ph.D. thesis,
  \bibinfo{school}{ETH Z\"{u}rich} (\bibinfo{year}{2005}),
  \urlprefix\url{http://arxiv.org/abs/quant-ph/0512258}.

\bibitem[{\citenamefont{Coles et~al.}(2012)\citenamefont{Coles, Colbeck, Yu,
  and Zwolak}}]{ColesColbeckYuZwolak2012PRL}
\bibinfo{author}{\bibfnamefont{P.~J.} \bibnamefont{Coles}},
  \bibinfo{author}{\bibfnamefont{R.}~\bibnamefont{Colbeck}},
  \bibinfo{author}{\bibfnamefont{L.}~\bibnamefont{Yu}}, \bibnamefont{and}
  \bibinfo{author}{\bibfnamefont{M.}~\bibnamefont{Zwolak}},
  \bibinfo{journal}{Phys. Rev. Lett.} \textbf{\bibinfo{volume}{108}},
  \bibinfo{pages}{210405} (\bibinfo{year}{2012}).

\bibitem[{\citenamefont{Coles}(2012)}]{ColesDecDisc2012}
\bibinfo{author}{\bibfnamefont{P.~J.} \bibnamefont{Coles}},
  \bibinfo{journal}{Phys. Rev. A} \textbf{\bibinfo{volume}{85}},
  \bibinfo{pages}{042103} (\bibinfo{year}{2012}).

\bibitem[{\citenamefont{Bhatia}(1997)}]{bhatia1997matrix}
\bibinfo{author}{\bibfnamefont{R.}~\bibnamefont{Bhatia}},
  \emph{\bibinfo{title}{Matrix analysis}}, vol. \bibinfo{volume}{169}
  (\bibinfo{publisher}{Springer}, \bibinfo{year}{1997}).

\bibitem[{\citenamefont{{Tomamichel} et~al.}(2012)\citenamefont{{Tomamichel},
  {Fehr}, {Kaniewski}, and {Wehner}}}]{TomEtAl2012arXiv1210.4359T}
\bibinfo{author}{\bibfnamefont{M.}~\bibnamefont{{Tomamichel}}},
  \bibinfo{author}{\bibfnamefont{S.}~\bibnamefont{{Fehr}}},
  \bibinfo{author}{\bibfnamefont{J.}~\bibnamefont{{Kaniewski}}},
  \bibnamefont{and} \bibinfo{author}{\bibfnamefont{S.}~\bibnamefont{{Wehner}}},
  \bibinfo{journal}{ArXiv e-prints}  (\bibinfo{year}{2012}),
  \eprint{1210.4359}.

\bibitem[{\citenamefont{{Schaffner}}(2007)}]{SchaffnerThesis2007}
\bibinfo{author}{\bibfnamefont{C.}~\bibnamefont{{Schaffner}}}, Ph.D. thesis,
  \bibinfo{school}{University of Aarhus} (\bibinfo{year}{2007}),
  \urlprefix\url{http://arxiv.org/abs/0709.0289}.

\bibitem[{\citenamefont{{Johnston}}()}]{NJohnstonPrivComm}
\bibinfo{author}{\bibfnamefont{N.}~\bibnamefont{{Johnston}}},
  \bibinfo{note}{private communication.}

\bibitem[{\citenamefont{Humphreys and Humphreys}(1975)}]{humphreys1975linear}
\bibinfo{author}{\bibfnamefont{J.~E.} \bibnamefont{Humphreys}}
  \bibnamefont{and} \bibinfo{author}{\bibfnamefont{J.~E.}
  \bibnamefont{Humphreys}}, \emph{\bibinfo{title}{Linear algebraic groups}},
  vol. \bibinfo{volume}{430} (\bibinfo{publisher}{Springer New York},
  \bibinfo{year}{1975}).

\bibitem[{\citenamefont{{Wehner} and {Winter}}(2010)}]{EURreview1}
\bibinfo{author}{\bibfnamefont{S.}~\bibnamefont{{Wehner}}} \bibnamefont{and}
  \bibinfo{author}{\bibfnamefont{A.}~\bibnamefont{{Winter}}},
  \bibinfo{journal}{New J. Phys.} \textbf{\bibinfo{volume}{12}},
  \bibinfo{pages}{025009} (\bibinfo{year}{2010}).

\bibitem[{\citenamefont{Deutsch}(1983)}]{deutsch}
\bibinfo{author}{\bibfnamefont{D.}~\bibnamefont{Deutsch}},
  \bibinfo{journal}{Physical Review Letters} \textbf{\bibinfo{volume}{50}},
  \bibinfo{pages}{631} (\bibinfo{year}{1983}).

\bibitem[{\citenamefont{{Bialynicki-Birula} and {Rudnicki}}()}]{EURreview2}
\bibinfo{author}{\bibfnamefont{I.}~\bibnamefont{{Bialynicki-Birula}}}
  \bibnamefont{and}
  \bibinfo{author}{\bibfnamefont{L.}~\bibnamefont{{Rudnicki}}},
  \emph{\bibinfo{title}{Entropic uncertainty relations in quantum physics}},
  \bibinfo{note}{e-print arXiv:1001.4668 [quant-ph]}.

\end{thebibliography}

\end{document}